\documentclass[submission,copyright,creativecommons]{eptcs}
\providecommand{\event}{GandALF 2016} % Name of the event you are submitting to
\usepackage{breakurl}             % Not needed if you use pdflatex only.

\usepackage[utf8]{inputenc}
\usepackage[T1]{fontenc}

\usepackage{xspace}
\usepackage{wrapfig}

\usepackage{amsmath}
\usepackage{amsthm}
\usepackage{amssymb}
\usepackage{mathtools}
\usepackage{nicefrac}

\theoremstyle{plain}
\newtheorem{theorem}{Theorem}
\newtheorem{lemma}{Lemma}

\usepackage{hyperref}
\usepackage{enumerate}
\usepackage{booktabs}

\usepackage{tikz}
\usepackage{xcolor}
\usepackage{pgfplots}
\usetikzlibrary{automata,calc,patterns,decorations,shapes}
\usetikzlibrary{decorations.pathreplacing}
\usetikzlibrary{decorations.pathmorphing}
\usetikzlibrary{positioning}
\tikzset{initial text={}}

\makeatletter
\pgfarrowsdeclare{stealthnew}{stealthnew}
{
  \ifdim\pgfgetarrowoptions{stealthnew}=-1pt%
    \pgfutil@tempdima=0.28pt%
    \pgfutil@tempdimb=\pgflinewidth%
    \ifdim\pgfinnerlinewidth>0pt%
      \pgfmathsetlength\pgfutil@tempdimb{.6\pgflinewidth-.4*\pgfinnerlinewidth}%
    \fi%
    \advance\pgfutil@tempdima by.3\pgfutil@tempdimb%
  \else%
    \pgfutil@tempdima=\pgfgetarrowoptions{stealthnew}%
    \divide\pgfutil@tempdima by 8%
  \fi%
  \pgfarrowsleftextend{+-3\pgfutil@tempdima}
  \pgfarrowsrightextend{+5\pgfutil@tempdima}
}
{
  \ifdim\pgfgetarrowoptions{stealthnew}=-1pt%
    \pgfutil@tempdima=0.28pt%
    \pgfutil@tempdimb=\pgflinewidth%
    \ifdim\pgfinnerlinewidth>0pt%
      \pgfmathsetlength\pgfutil@tempdimb{.6\pgflinewidth-.4*\pgfinnerlinewidth}%
    \fi%
    \advance\pgfutil@tempdima by.3\pgfutil@tempdimb%
  \else%
    \pgfutil@tempdima=\pgfgetarrowoptions{stealthnew}%
    \divide\pgfutil@tempdima by 8%
    \pgfsetlinewidth{0bp}%
  \fi%
  \pgfpathmoveto{\pgfqpoint{5\pgfutil@tempdima}{0pt}}
  \pgfpathlineto{\pgfqpoint{-3\pgfutil@tempdima}{4\pgfutil@tempdima}}
  \pgfpathlineto{\pgfpointorigin}
  \pgfpathlineto{\pgfqpoint{-3\pgfutil@tempdima}{-4\pgfutil@tempdima}}
  \pgfusepathqfill
}
\makeatother

\pgfsetarrowoptions{stealthnew}{-1pt}
\pgfkeys{/tikz/.cd, arrowhead/.default=-1pt, arrowhead/.code={
  \pgfsetarrowoptions{stealthnew}{#1},
}}

\newcommand{\myquot}[1]{``#1''}
\newcommand{\halfthinspace}{{\kern .08333em}}
\newcommand{\bosy}{\texttt{BoSy}\xspace}

\newcommand{\qbf}{QBF\xspace}

\newcommand{\coloneq}{\mathop{:=}}

\newcommand{\cceq}{\mathop{::=}}

\newcommand{\bigo}{\mathcal{O}}
\newcommand{\smallo}{\ensuremath{o}}

\renewcommand{\epsilon}{\varepsilon}
\renewcommand{\phi}{\varphi}

\newcommand{\pow}[1]{2^{#1}}
\newcommand{\nats}{\mathbb{N}}
\newcommand{\card}[1]{|#1|}
\newcommand{\size}[1]{\card{#1}}
\newcommand{\set}[1]{\{ #1 \}}
\newcommand{\opt}{\mathrm{opt}}
\newcommand{\aut}{\mathfrak{A}}

\newcommand{\game}{\mathcal{G}}

\newcommand{\mem}{\mathcal{M}}

\newcommand{\update}{\mathrm{upd}}
\newcommand{\nextmove}{\mathrm{nxt}}

\newcommand{\true}{\mathbf{tt}}
\newcommand{\false}{\mathbf{ff}}
\newcommand{\F}{{\mathbf{F\,}}}
\newcommand{\G}{{\mathbf{G\,}}}
\newcommand{\Fp}{{\mathbf{F_P\,}}}
\newcommand{\U}{{\mathbf{\,U\,}}}
\newcommand{\X}{{\mathbf{X\,}}}
\newcommand{\R}{{\mathbf{\,R\,}}}

\newcommand{\rel}[0]{\mathrm{rel}}
\newcommand{\realparms}{{\mathcal R}}

\newcommand{\ltl}{\mathrm{LTL}}

\newcommand{\pltl}{\mathrm{PLTL}}
\newcommand{\prompt}{\mathrm{PROMPT}$\textendash$\ltl}

\newcommand{\pldl}{\text{PLDL}}
\newcommand{\ldl}{\text{LDL}}

\newcommand{\pspace}{{\textsc{Pspace}}}
\newcommand{\twoexp}{{\textsc{2Exptime}}}
\newcommand{\threeexp}{{\textsc{3Exptime}}}

\newcommand{\imark}{\ensuremath{i}\xspace}
\newcommand{\omark}{\ensuremath{o}\xspace}
\newcommand{\idelim}{\ensuremath{\#_i}\xspace}
\newcommand{\odelim}{\ensuremath{\#_o}\xspace}
	\newcommand{\pick}{\mathrm{Pick}}

\title{Approximating Optimal Bounds in Prompt-LTL Realizability in Doubly-exponential Time\thanks{Supported by the projects TriCS (ZI 1516/1-1) and AVACS (SFB/TR 14) of the German Research Foundation (DFG) and by the European Research Council (ERC) Grant OSARES (No.\ 683300)}}
\author{Leander Tentrup, Alexander Weinert, and Martin Zimmermann
\institute{Reactive Systems Group, Saarland University, 66123 Saarbrücken, Germany}
\email{\{tentrup, weinert, zimmermann\}@react.uni-saarland.de}}

\def\titlerunning{Approximating Optimal Bounds in Prompt-LTL Realizability in Doubly-exponential Time}
\def\authorrunning{L. Tentrup, A.Weinert and M. Zimmermann}
\begin{document}

\maketitle

%%%%%%%%%%%%%%%%%%%%%%%%%%%%%%%%%%%%%%%%%%%%%%%%%%%%%
%%%%%%%%%%%%%%%%%%%%%%%%%%%%%%%%%%%%%%%%%%%%%%%%%%%%%

\begin{abstract}
We consider the optimization variant of the realizability problem for Prompt Linear Temporal Logic, an extension of Linear Temporal Logic (LTL) by the prompt eventually operator whose scope is bounded by some parameter. In the realizability optimization problem, one is interested in computing the minimal such bound  that allows to realize a given specification. It is known that this problem is solvable in triply-exponential time, but not whether it can be done in doubly-exponential time, i.e., whether it is just as hard as solving LTL realizability. 

We take a step towards resolving this problem by showing that the optimum can be approximated within a factor of two in doubly-exponential time. Also, we report on a proof-of-concept implementation of the algorithm based on bounded LTL synthesis, which computes the smallest implementation of a given specification. In our experiments, we observe a tradeoff between the size of the implementation and the bound it realizes. We investigate this tradeoff in the general case and prove upper bounds, which reduce the search space for the algorithm, and matching lower bounds.
\end{abstract}

%%%%%%%%%%%%%%%%%%%%%%%%%%%%%%%%%%%%%%%%%%%%%%%%%%%%%
%%%%%%%%%%%%%%%%%%%%%%%%%%%%%%%%%%%%%%%%%%%%%%%%%%%%%
\section{Introduction}
\label{section_intro}
The realizability problem for $\prompt$, Linear Temporal Logic ($\ltl$) enriched with an eventually operator of bounded scope, should be treated as an optimization problem: determine the smallest bound on the bounded eventually such that the specification is realizable with respect to that bound. The best exact algorithms for this problem have triply-exponential running times, i.e., they are exponentially slower than algorithms for the decision variant (\myquot{does there exist a bound?}), which is $\twoexp$-complete. We take a step towards resolving the complexity of the optimization problem by presenting an approximation algorithm with doubly-exponential running time  returning a bound that is at most twice the optimum. 

In general, the realizability problem asks to determine the winner in an infinite-duration two-player game played between an input and an output player in rounds~$n = 0,1,2, \ldots$: in each round~$n$, first the input player picks a subset~$i_n$ of a fixed set $I$ of input propositions, then the output player picks a subset~$o_n$ of a fixed set $O$ of output propositions. The output player wins, if the sequence~$(i_0 \cup o_0)
(i_1 \cup o_1)
(i_2 \cup o_2)
\cdots$ of picks satisfies the winning condition, typically a formula~$\phi$ in some logic. A strategy for the output player is a function mapping sequences~$i_0 \cdots i_n \in (\pow{I})^*$ of inputs to an output~$o_n \in \pow{O}$. Such a strategy is winning, if every outcome that is consistent with the strategy satisfies the winning condition. Formally, the realizability problem asks, given a formula~$\phi$, whether the output player has a winning strategy for the realizability game with winning condition~$\phi$. For winning conditions in $\ltl$ (and many extensions), finite-state strategies suffice, i.e., strategies that are implemented by finite automata with outputs. 

$\ltl$~\cite{Pnueli77} is the most prominent logic for specifying reactive systems and the foundations of the $\ltl$ realizability problem are well-understood~\cite{AbadiLamportWolper89,KupfermanPitermanVardi06,KupfermanVardi05,PnueliRosner89,PnueliRosner89a}. Recently, the first tools solving the problem were developed~\cite{BohyBruyereFiliotJinRaskin,Ehlers11e,FiliotJinRaskin11,FinkbeinerSchewe13,JobstmannBloem06}, which show promising performance despite the prohibitive worst-case complexity. However, $\ltl$ lacks the ability to express time-bounds, e.g., the formula~$\G(q \rightarrow \F p)$ expresses that every request~$q$ has to be responded to by a response~$p$. However, it does \emph{not} require a bound on the waiting times between requests and responses, i.e., it is even satisfied if the waiting times diverge. Several parameterized logics where introduced to overcome this shortcoming~\cite{AlurEtessamiLaTorrePeled01,FaymonvilleZimmermann14,KupfermanPitermanVardi09,Zimmermann15a}. Here, we focus on the smallest such logic: $\prompt$, which extends $\ltl$ by the prompt eventually operator~$\Fp$, whose semantics are defined with respect to a given bound~$k$. For example, the formula ~$\G(q \rightarrow \Fp p)$ is satisfied with respect to $k$, if every request is responded to within at most $k$ steps. In decision problems for this logic the bound is typically quantified existentially, e.g., the realizability problem asks for a given formula~$\phi$ whether there exists a bound~$k$ such that the output player has a winning strategy for the realizability game where the winning condition~$\phi$ is evaluated with respect to $k$.

Kupferman et al.\ showed that $\prompt$ has the same desirable algorithmic properties as $\ltl$. In particular, model checking is $\pspace$-complete and realizability is $\twoexp$-complete~\cite{KupfermanPitermanVardi09}. Hence, one can add the prompt eventually operator to $\ltl$ for free. However, as already noticed by Alur et al.\ in their work on Parametric $\ltl$~\cite{AlurEtessamiLaTorrePeled01} (which also contains the dual of the prompt eventually and allows for multiple bounds), one can view decision problems for parameterized logics as optimization problems: instead of asking for the existence of some bound, one searches for an optimal one. They showed that the model checking optimization problem for unipolar $\pltl$ specifications, which includes $\prompt$, can be solved in polynomial space~\cite{AlurEtessamiLaTorrePeled01}. Thus, even finding optimal bounds is not harder than solving the $\ltl$ model checking problem. However, for $\prompt$ realizability, or equivalently, for infinite games, the situation is different: while the decision problem is known to be $\twoexp$-complete~\cite{KupfermanPitermanVardi09}, the best algorithm for the optimization problem has triply-exponential running time~\cite{Zimmermann13}. 
 
\subsection{Our Contributions}

We show that relaxing the optimality requirement on the bound allows to recover doubly-exponential running times: an approximately optimal bound can be determined using the alternating color technique, which was introduced by Kupferman et al.\ to solve the decision problems for $\prompt$. To this end, we present an approximation algorithm with doubly-exponential running time with an approximation ratio of two. The algorithm has to solve at most doubly-exponentially many $\ltl$ realizability problems, each solvable in doubly-exponential time. We present the algorithm for $\prompt$, but it is applicable to stronger parameterized extensions of $\ltl$ like parametric $\ltl$~\cite{AlurEtessamiLaTorrePeled01} and parametric $\ldl$~\cite{FaymonvilleZimmermann14}. 

In many situations, approximating the optimal bound is sufficient, since the exact optimum depends on the granularity of the realizability problem at hand. This is even more true if the optimization problem indeed turns out to be harder than the decision variant, e.g., if it is $\threeexp$-hard. Then, the loss in quality is made up for by significant savings in running time. On the other hand, if the optimal bound is at most exponential in the size of the formula, then it can be exactly determined in doubly-exponential time~\cite{Zimmermann13}: the bound can be hardwired into a non-deterministic automaton capturing the specification, which has to be determinized to solve the realizability problem. This involves an exponential blow-up, which implies that this approach only yields a doubly-exponential time algorithm, if the bound is at most exponential. 

Furthermore, we report on a proof-of-concept implementation of our algorithm. To handle the solution of the $\ltl$ realizability problems, we rely on the framework of bounded synthesis~\cite{FinkbeinerSchewe13}, which searches for a minimal-size finite-state winning strategy for a given specification. The evaluation of this implementation shows that, while it suffers from a significant increase in running time compared to $\ltl$ realizability, synthesis of prompt arbiters for a small number of clients is feasible.

In our experiments, a tradeoff between size and quality (measured in the bound on the prompt eventually operators) of winning strategies becomes apparent: one can trade size of the strategy for quality and vice versa. We conclude by studying this tradeoff in depth. First, we show that fixing the size of the strategy to $n$ (as it is done during bounded synthesis) implies an exponential upper bound (in $n$) on the sufficient bound~$k$ on the prompt eventually operators. This upper bound reduces the search space of our algorithm. The upper bound is then matched by a tight lower bound. Secondly, we present a family of formulas exhibiting a continuous tradeoff between size and quality with exponential extremal values, i.e., the specifications are realizable with exponential size and a linear bound or with constant size and an exponential bound and the tradeoff between these two points is continuous. Thirdly, by giving up the continuity, one can show even stronger tradeoffs: there is a family of specifications that is realizable with doubly-exponential size and bound zero or with size one and an exponential bound. 

%%%%%%%%%%%%%%%%%%%%%%%%%%%%%%%%%%%%%%%%%%%%%%%%%%%%%
%%%%%%%%%%%%%%%%%%%%%%%%%%%%%%%%%%%%%%%%%%%%%%%%%%%%%
\section{Definitions}
\label{section_defs}
Throughout this work, fix a finite set~$P$ of atomic propositions and denote the non-negative integers by~$\nats$.

\subsection{Prompt-LTL}
The formulas of $\prompt$ are given by the grammar
\begin{equation*}\phi \cceq p \mid \neg p \mid \phi \wedge \phi \mid \phi \vee
\phi \mid \X \phi \mid \phi \U \phi \mid \phi \R \phi \mid \Fp \phi
  ,\end{equation*}
where $p \in P$ represents an atomic proposition. Also, we use the standard shorthands~$\F \phi = \true \U \phi$ and $\G \phi = \false \R \phi$ with $\true = p \vee \neg p$ and $\false = p \wedge \neg p$, where $p$ is a fixed atomic proposition.
Furthermore, we use $\phi \rightarrow \psi$ as shorthand for $\neg \phi \vee \psi$,
if the antecedent~$\phi$ is a (negated) atomic
proposition (where we identify $\neg \neg a$ with $a$). 
We define the size~$\size{\phi}$ of $\phi$ to be the number of subformulas of $\phi$.

In order to evaluate $\prompt$ formulas, we need to fix a bound~$k \in \nats$ to evaluate the prompt eventually operator. Hence, the satisfaction 
relation is defined for an $\omega$-word~$w \in \left( \pow{P} \right)^{ \omega }$, a
position~$n$ of $w$, a bound~$k$, and a~$\prompt$ formula. The definition is standard for the classical operators and defined as follows for the prompt eventually:
\begin{itemize}

\item $(w,n,k)\models\Fp\phi$ if and only if there exists a $j$ in
the range~$0\le j \le k$ such that $(w,n+j,k)\models\phi$.

\end{itemize}
For the sake of brevity, we write $(w,k) \models \phi$ instead of
$(w,0,k) \models \phi$ and say that $w$ is a model of $\phi$ with
respect to $k$.
If $(w, k) \models \phi$, we say that $w$ models $\phi$ with respect to $k$.
Note that $\phi$ is an $\ltl$ formula~\cite{Pnueli77}, if it does not contain the prompt eventually. In this case, we write $w \models \phi$.

\subsection{Prompt-LTL Realizability}
\label{subsection_realizability}
Throughout this subsection, we fix a partition~$(I,O)$ of $P$. An instance of the $\prompt$ realizability problem over $(I,O)$ consists of an $\prompt$ formula~$\phi$ over $P = I \cup O$ and asks to determine the winner in the following game, played between Player~$I$ and Player~$O$ in rounds~$n = 0, 1, 2, \ldots$: in round~$n$, Player~$I$ picks $i_n \subseteq I$ and afterwards Player~$O$ picks $o_n \subseteq O$. The resulting play is $(i_0 \cup o_0)
(i_1 \cup o_1)
(i_2 \cup o_2)
\cdots \in (\pow{P})^\omega$. 

A strategy for Player~$O$ is a mapping $\sigma \colon (\pow{I})^+ \rightarrow \pow{O}$. A play as above is consistent with $\sigma$, if $o_n = \sigma(i_0 \cdots i_n)$ for every $n$. We say that $\sigma$ realizes $\phi$  with respect to $k \in \nats$, if every play that is consistent with $\sigma$ satisfies $\phi$ with respect to $k$. Formally, the $\prompt$ realizability problem asks, given a $\prompt$ formula~$\phi$, whether there is a strategy $\sigma$ and a $k$ such that $\sigma$ realizes $\phi$ with respect to $k$. In this case, we say $\phi$ is realizable.

A memory structure $\mem = (M, m_0, \update)$ consists of a finite set of states $M$, an initial state $m_0 \in M$, and an update function $\update\colon M \times 2^I \rightarrow M$.
We extend the update function to finite input sequences as usual, i.e., we define $\update^* \colon (2^I)^* \rightarrow M$ inductively as $\update^*(\epsilon) = m_0$ and $\update^*(wi) = \update(\update^*(w), i)$ for $w \in (\pow{I})^*$ and $i \in \pow{I}$.
A memory structure $\mem$ together with a next-move function $\nextmove\colon M \times 2^I \rightarrow 2^O$ induces a strategy $\sigma$ defined as $\sigma(i_0\cdots i_n) = \nextmove(\update^*(i_0 \cdots i_{n-1}), i_n)$.
We say that such a memory structure implements the strategy $\sigma$.
We call any strategy $\sigma$ that can be implemented by some memory structure a finite-state strategy.
The size of a finite-state strategy is the size of the smallest memory structure implementing it.

The $\ltl$ realizability problem is defined by restricting the specifications~$\phi$ to
$\ltl$ formulas and is $\twoexp$-complete~\cite{PnueliRosner89a}.
Kupferman et al. showed that $\prompt$ realizability is not harder.

\begin{theorem}[\cite{KupfermanPitermanVardi09}]
\label{thm_prompt}
The $\prompt$ realizability problem is $\twoexp$-com\-plete. Furthermore, if $\phi$ is realizable with respect to some $k$, then also with respect to some $k \in \bigo(2^{2^{\size{\phi}}})$ by some finite-state strategy of size~$\bigo(2^{2^{\card{\phi}}})$.
\end{theorem}

Furthermore, the doubly-exponential upper bounds on the necessary $k$ and on the memory requirements are tight. Also, if $\phi$ is realizable with respect to some $k$, then also with respect to every $k'>k$.

\subsection{The Alternating Color Technique}
Our algorithm presented in the next section is based on an application of Kupferman et al.'s alternating color technique \cite{KupfermanPitermanVardi09} to $\prompt$ realizability. We recall the technique in this subsection.

Let $p\notin P$ be a fixed fresh proposition. An
$\omega$-word~$w'\in\left(2^{P\cup\{p\}}\right)^{\omega}$ is a $p$-coloring of
$w\in\left(2^{P}\right)^{\omega}$ if $w_n'\cap P=w_n$, i.e., $w_n$ and $w_n'$
coincide on all propositions in $P$. We say
that a position is a change point, if $n=0$ or if the truth value of $p$ at positions~${n-1}$
and $n$ differs. A $p$-block is an infix~$w_m' \cdots w_{n}'$ of $w'$ such that $m$ and $n+1$ are adjacent change points. Let $k \ge 1$: we say
that $w'$ is $k$-spaced, if the truth value of $p$ changes infinitely often and each
$p$-block has length at least $k$; we say that $w'$ is $k$-bounded, if each
$p$-block has length at most $k$ (which implies that the truth value of $p$ changes
infinitely often).

Given a $\prompt$ formula~$\phi$, $\rel(\phi)$ denotes the formula obtained by inductively replacing
every subformula~$\Fp\psi$ by
\begin{equation*}
(p\rightarrow (p\U(\neg p\U\rel(\psi))))\wedge(\neg p\rightarrow (\neg p\U(
p\U\rel(\psi)))).
\end{equation*}
Intuitively, instead of requiring $\psi$ to be satisfied within a bounded number of steps, $\rel(\phi)$ requires it to be satisfied within at most one change point. The relativization~$\rel(\phi)$ is an $\ltl$ formula of size~$\bigo(\card{\phi})$. Kupferman et al.\ showed that $\phi$ and $\rel(\phi)$ are ``equivalent'' on $\omega$-words
which are bounded and spaced.

\begin{lemma}[\cite{KupfermanPitermanVardi09}]
\label{lemma_alternatingcolor}
Let $\phi$ be a $\prompt$ formula.
\begin{enumerate}
\item \label{lemma_alternatingcolor_prompttoltl}
If $(w,k)\models \phi$, then $w' \models \rel(\phi)$ for
every $k$-spaced $p$-coloring~$w'$ of $w$.

\item \label{lemma_alternatingcolor_ltltoprompt}
Let $k\in\nats$. If $w'$ is a $k$-bounded $p$-coloring of $w$ such that
$w' \models \rel(\phi)$, then $(w,2k)\models\phi$.
\end{enumerate}
\end{lemma}

%%%%%%%%%%%%%%%%%%%%%%%%%%%%%%%%%%%%%%%%%%%%%%%%%%%%%
%%%%%%%%%%%%%%%%%%%%%%%%%%%%%%%%%%%%%%%%%%%%%%%%%%%%%
\section{Approximating Optimal Bounds in Doubly-Exponential Time}
\label{section_approx}
Determining whether a $\prompt$ formula~$\phi$ is realizable with respect to some $k$ induces a natural optimization problem: determine the smallest such $k$. The optimum (and a strategy realizing $\phi$ with respect to the optimum) can be computed in triply-exponential time~\cite{Zimmermann13}.

However, it is an open problem whether the optimization problem can be solved in doubly-exponential time, i.e., whether optimal $\prompt$ realizability is no harder than $\ltl$ realizability. We take a step towards resolving the problem by showing that the optimum can be approximated within a factor of two in doubly-exponential time. 

The alternating color technique is applied to the $\prompt$ realizability problem by replacing $\phi$ by its relativization~$\rel(\phi)$ and by letting Player~$O$ determine the truth value of the distinguished proposition~$p$ for every position by adding it to the output propositions~$O$. The full details are explained in~\cite{KupfermanPitermanVardi09}, where the following statements are shown to prove the application of the alternating color technique to be correct. Here, $\psi_k$ is an $\ltl$ formula of linear size in $k$ that characterizes $k$-boundedness, i.e., $w' \models \psi_k$ if, and only if, $w'$ is a $k$-bounded $p$-coloring. 

\begin{lemma}[\cite{KupfermanPitermanVardi09}]
\label{lemma_altcolorgames}
Let $\phi$ be a $\prompt$ formula and let $k \in \nats$.
\begin{enumerate}

\item\label{lemma_altcolorgames_prompt2ltl}
 A strategy realizing $\phi$ with respect to $k$ can be turned into a strategy realizing $\rel(\phi) \wedge \psi_k$.

\item\label{lemma_altcolorgames_ltl2prompt}
A strategy realizing $\rel(\phi) \wedge \psi_k$ can be turned into a strategy realizing $\phi$ with respect to $2k$.

\end{enumerate}
\end{lemma}

Also, if $k$ is not \emph{too large}, we can check the realizability of $\rel(\phi) \wedge \psi_k$ in doubly-exponential time.

\begin{lemma}
\label{lemma_altcolorgames_complexity}
 The following problem is in $\twoexp$: Given a $\prompt$ formula~$\phi$ and a natural number~$k$ that is at most doubly-exponential in $\size{\phi}$, is $\rel(\phi) \wedge \psi_k$ realizable? Furthermore, one can compute a strategy realizing the formula (if one exists) in doubly-exponential time.
\end{lemma}

\begin{proof}
As usual, we reduce the problem to a parity game (see \cite{GraedelThomasWilke02} for background). First, we construct a deterministic parity automaton recognizing the language~$\set{ \rho \in (\pow{P \cup \set{p}})^\omega \mid \rho \models \rel(\phi)}$ and intersect it with a deterministic safety automaton that recognizes 
$\set{ \rho \in (\pow{P \cup \set{p}})^\omega \mid \text{$\rho \models \psi_k$}}$.
It is known that the first automaton is of doubly-exponential size and has exponentially many colors (both in $\size{\phi}$) while the second one is of linear size in $k$. Thus, the deterministic parity automaton~$\aut$ recognizing the intersection is of doubly-exponential size in $\size{\phi}$ and linear size in $k$ and has exponentially many colors in $\size{\phi}$. 

Next, we split a transition of $\aut$ labeled by $A \subseteq P \cup \set{p}$ into two, the first one labeled by $A \cap I$ and the second one by $A \setminus I$. By declaring the original states of $\aut$ to be Player~$I$ states and the new intermediate states obtained by splitting the transitions to be Player~$O$ states, we obtain a parity game that is won by Player~$O$ from the initial state of $\aut$ if, and only if, $\rel(\phi) \wedge \psi_k$ is realizable. Additionally, a winning strategy for Player~$O$ in the parity game can be turned into a strategy realizing $\rel(\phi) \wedge \psi_k$.
This parity game is of doubly-exponential size with exponentially many colors, both in $\size{\phi}$. The winner and a winning strategy for her in such a game can be computed in doubly-exponential time~\cite{Schewe07}. 
\end{proof}

Now, we are able to present the algorithm for approximating optimal bounds for $\prompt$ realizability. Given an input $\phi$, the algorithm first checks whether $\phi$ is realizable with respect to some $k$. If not, then the optimum is $\infty$ by convention. Otherwise, Theorem~\ref{thm_prompt} yields a doubly-exponential upper bound~$u$ on the optimum. Now, the algorithm determines the smallest~$1 \le k \le u$ such that $\rel(\phi) \wedge \psi_k$ is realizable and returns $2k$. The emptiness test and determining the realizability of $\rel(\phi) \wedge \psi_k$ can be executed in doubly-exponential time as shown in Theorem~\ref{thm_prompt} and Lemma~\ref{lemma_altcolorgames_complexity}. As the latter problem has to be solved at most doubly-exponentially often\footnote{With binary search, this can be improved to exponentially often. However, the running time of the realizability check depends on $k$, which is typically small. Thus, traversing the search space $0,1,\ldots,u$ in the natural order is more beneficial. We discuss the search strategy in more detail in Section~\ref{section_experiments}.}, the overall running time is doubly-exponential as well. Furthermore, due to Lemma~\ref{lemma_altcolorgames_complexity} and Lemma~\ref{lemma_altcolorgames}.\ref{lemma_altcolorgames_ltl2prompt}, we even obtain a strategy realizing $\phi$ with respect to $2k$.

It remains to argue that the algorithm approximates the optimum~$k_\opt \le u$ within a factor of two: let $2k$ be the output of the approximation algorithm, i.e., $k$ is minimal such that $\rel(\phi) \wedge \psi_k$ is realizable. Thus, Lemma~\ref{lemma_altcolorgames}.\ref{lemma_altcolorgames_ltl2prompt} implies $k_\opt \le 2k$. Conversely, $\phi$ being realizable with respect to $k_\opt$ implies that  $\rel(\phi) \wedge \psi_{k_\opt}$ is realizable due to Lemma~\ref{lemma_altcolorgames}.\ref{lemma_altcolorgames_prompt2ltl}, i.e., $k \le k_\opt$ due to minimality of $k$. 

Altogether, we obtain $k \le k_\opt \le 2k$. Recall that the algorithm returns $2k$, i.e., $\phi$ is realizable with respect to the returned value due to monotonicity. Also, the approximation ratio~$\nicefrac{2k}{2k - k_\opt}$ is bounded by 
$
\nicefrac{2k}{2k - k_\opt} \le \nicefrac{2k}{2k - k} = 2$, i.e., the bound found by our algorithm is at most twice the optimal bound.

\begin{theorem}
The optimization problem for $\prompt$ realizability can be approximated within a factor of two in doubly-exponential time. As a byproduct, one obtains a strategy witnessing the approximatively optimal bound.
\end{theorem}

%%%%%%%%%%%%%%%%%%%%%%%%%%%%%%%%%%%%%%%%%%%%%%%%%%%%%
%%%%%%%%%%%%%%%%%%%%%%%%%%%%%%%%%%%%%%%%%%%%%%%%%%%%%
\section{Empirical Evaluation}
\label{section_experiments}
In the previous section we have described an algorithm that, given some $\prompt$ specification $\varphi$, approximates the optimal bound~$k$ for which the formula can be realized. The algorithm uses $\ltl$ realizability checking as a black-box to determine the realizability of the formulas~$\rel(\phi) \land \psi_k$, where $k$ is a parameter from a doubly-exponential set. The search strategy heavily influences the running time of the algorithm (but not the worst-case complexity). 
%The method for finding such a $k$ is very basic, however.
%Since it is known that, if any $k$ exists with respect to which $\varphi$ is realizable, $\varphi$ is also realizable with respect to some fixed $u$ that is doubly-exponential in $\size{\varphi}$, it suffices to try out all $k$ with $1 \leq k \leq u$.
%For any given $k$, the problem is then reduced to checking realizability of the $\ltl$ formula~$\rel(\phi) \land \psi_k$. 
%Thus, we use $\ltl$ realizability checking as a black-box in our algorithm.
Towards an implementation, we rely on bounded $\ltl$ synthesis \cite{FinkbeinerSchewe13} for checking the realizability of $\rel(\varphi) \land \psi_k$.
In addition to computing the smallest strategy that realizes $\rel(\varphi) \land \psi_k$, bounded synthesis also allows us to search for strategies of some fixed size $n$.
Thus, we obtain a sub-procedure that takes as input some $\prompt$ formula, as well as some values $n$ and $k$, which checks whether or not there exists a finite-state strategy of size $n$ that realizes $\varphi$ with respect to $2k$.

To this end, it first constructs the $\ltl$ formula $\varphi' = \rel(\phi) \land \psi_k$ from $\varphi$, which is then given to the tool \bosy~\cite{FaymonvilleFinkbeinerRabeTentrup13} together with the desired size $n$ of the strategy.
\bosy then checks $\varphi'$ for realizability and returns a strategy of size $n$, if there exists one.
In order to do so, it first translates $\varphi'$ to a universal co-B\"uchi automaton that accepts the language of $\varphi'$. Based on this automaton, it constructs a \qbf query that is satisfiable if, and only if, there exists a strategy of size $n$ which is then solved by a combination of a \qbf preprocessor and a solver.
Due to Lemma~\ref{lemma_alternatingcolor_ltltoprompt}, we know that the strategy returned by \bosy realizes $\varphi$ with respect to $2k$ when restricted to $P$.

We evaluate our implementation on a family of arbiters.
Each arbiter manages some number $r$ of resources.
Player~$I$ poses requests $q_i$ for some resource~$1 \leq i \leq r$, while Player~$O$ has to grant them by playing $p_i$ for $1 \leq i \leq r$.
Moreover, Player~$O$ can only grant a single resource at a time.
In addition to the usual requirement that each request has to be answered eventually, we require that a request for one of the first $r_p$ resources is answered promptly, for $0 \leq r_p \leq r$.
Thus, for some parameters~$r$ and $r_p$, we construct the $\prompt$ specification
\[ \varphi_{r, r_p} \coloneq \bigwedge_{1 \leq i \leq r_p} \G(q_i \rightarrow \Fp p_i) \land \bigwedge_{r_p < i \leq r} \G (q_i \rightarrow \F p_i) \land \bigwedge_{i \neq j} \G (\neg p_i \lor \neg p_j). \]
Note that, for each $r \in \nats$, the specification $\varphi_{r, 0}$ is an $\ltl$ formula.

For our experiments we used machines equipped with Intel Xeon-Haswell processors running at 3.6\thinspace GHz with 32\thinspace GB of memory.
The complete dataset we report on in this evaluation is available at \url{https://arxiv.org/abs/1511.09450}.

We first compare the running time of $\ltl$ synthesis with the running time of our implementation on the $\prompt$ formulas in order to quantify the slowdown incurred by performing $\prompt$ synthesis instead of $\ltl$ synthesis.
Since, as previously explained, a naive search strategy that simply performs bounded synthesis on $\rel(\varphi) \land \psi_k$ for increasing $k$ is infeasible, we instead search for a realizing implementation along the diagonals of the search space, as shown in the left-hand side of Figure~\ref{fig:search_strategy}.
We run our implementation with this search strategy on $\varphi_{r,r_p}$ for each $r \in [1;10]$ and each $r_p \in [0;r]$ and compared the running time to that of \bosy on $\varphi_{r,0}$.\footnote{Note that it is not possible to run \bosy on $\varphi_{r, r_p}$ for $r_p > 0$, as \bosy performs $\ltl$ synthesis, while $\varphi_{r, r_p}$ is a $\prompt$ formula for $r_p > 0$.}

The results are shown in the right-hand side of Figure~\ref{fig:search_strategy}.
For each comparison, the number of resources $r$ is denoted by the line-color, while the number of prioritized resources is displayed on the x-axis.
The slowdown is shown on the y-axis, which is logarithmically scaled.
Note that there does not exist a data point for each pair $(r, r_p) $ with $1 \leq r \leq 10$ and $0 \leq r_p \leq r$, since, for $r \geq 9$, \bosy timed out after 100 minutes and, for all other values not shown, our implementation timed out after 100 minutes.

\begin{figure}
	\centering
	\raisebox{.6cm}{\begin{tikzpicture}[thick,-stealthnew,arrowhead=3mm]
		% Use - in order to remove arrow tip from grid
		\draw[thin,gray,-] (0,0) grid (4.5,4.5);
		\path[draw,-stealth,gray,-stealth] (0,0) -- (0,4.5);
		\path[draw,-stealth,gray,-stealth] (0,0) -- (4.5,0);
		
		\foreach \coord in {1,2,3,4,5}
			\node[anchor=north] at ($(\coord,0) - (1,0)$) {$\coord$};
		\foreach \coord in {1,2,3,4,5}
			\node[anchor=east] at ($(0,\coord) - (0,1)$) {$\coord$};
		
		\path[draw] (0,0) -- (0,1);
		\path[draw] (0,1) -- (1,0);
		\path[draw] (1,0) -- (0,2);
		\path[draw] (0,2) -- (1,1);
		\path[draw] (1,1) -- (2,0);
		\path[draw] (2,0) -- (0,3);
		\path[draw] (0,3) -- (1,2);
		\path[draw] (1,2) -- (2,1);
		\path[draw] (2,1) -- (3,0);
		\path[draw] (3,0) -- (0,4);
		\path[draw] (0,4) -- (1,3);
		\path[draw] (1,3) -- (2,2);
		\path[draw] (2,2) -- (3,1);
		\path[draw] (3,1) -- (4,0);
		\path[draw,dashed] (4,0) -- (.33,4.5);
		
		\node[anchor=north west] at (4.5,0) {$k$};
		\node[anchor=south east] at (0,4.5) {$n$};
	\end{tikzpicture}}
	\hfill
	\begin{tikzpicture}[thick,xscale=.9,yscale=.8]
		\begin{axis}[
			thick,
			ymode=log,
			xlabel=$r_p$,
			ylabel={Slowdown},
			grid=both,
			enlarge x limits=true,
			legend pos = outer north east,
			legend cell align = left
		]
			\addplot+[] table {data/slowdown-1.dat};
			\addplot+[] table {data/slowdown-2.dat};
			\addplot+[] table {data/slowdown-3.dat};
			\addplot+[] table {data/slowdown-4.dat};
			\addplot+[] table {data/slowdown-5.dat};
			\addplot+[] table {data/slowdown-6.dat};
			\addplot+[] table {data/slowdown-7.dat};
			\addplot+[] table {data/slowdown-8.dat};
			\addplot+[] table {data/slowdown-9.dat};
			\addplot+[] table {data/slowdown-10.dat};
		\legend{$r = 1$, $r=2$, $r=3$, $r=4$, $r=5$, $r=6$, $r=7$, $r=8$, $r=9$, $r=10$} 
		\end{axis}
	\end{tikzpicture}
	
	\caption{The search strategy for some realizing implementation on the left-hand side and the slowdown of $\prompt$ synthesis on the right-hand side.}
	\label{fig:search_strategy}
\end{figure}
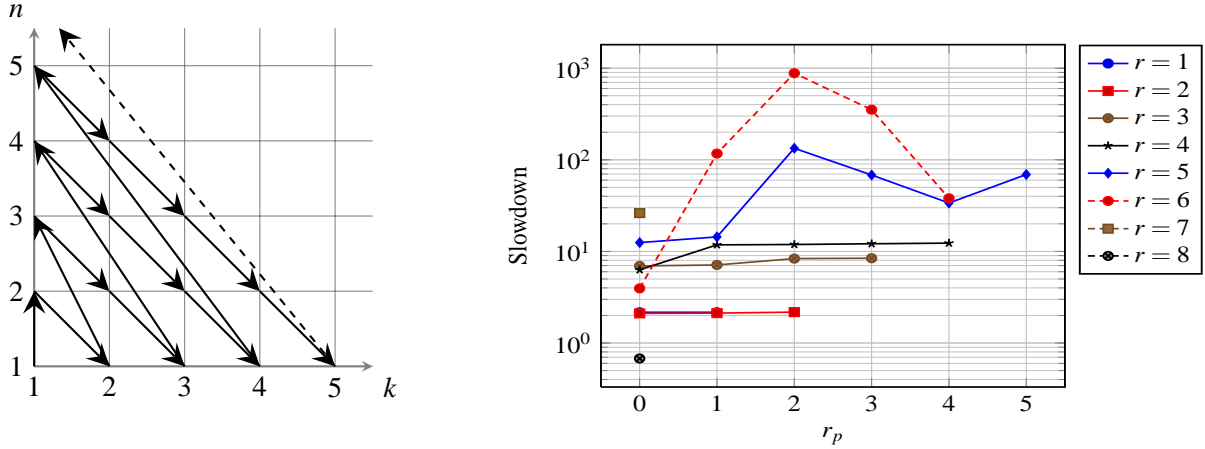

We see that, when given an $\ltl$ formula $\varphi_{r,0}$, in general our implementation is slower than \bosy by a factor on the order of magnitude of $10^1$.
This results from our tool calling \bosy multiple times even for $\ltl$ formulas, as, in order to find a strategy of size $n$ that realizes $\varphi$, our implementation first searches for strategies of sizes $n'-b$ that realize $\varphi_{r,0}$ with respect to $1 + b$ for $n' < n$ and $0 \leq b < n'$ (cf.\ the search strategy shown in Figure~\ref{fig:search_strategy}).
For $\varphi_{8,0}$, however, our implementation finds a realizing strategy after 1\thinspace 299 seconds, while \bosy takes 1\thinspace 914 seconds for the same task.
This discrepancy is likely due to differences in the generated automaton that lead to different QBF formulas and result in different solving times.
%We assume that this difference is due to the fact that the \bosy invocation from our implementation specifies a concrete size of the strategy to be synthesized, while without such a parameter, \bosy searches iteratively.

When asked to realize $\varphi_{r,r_p}$ with $r_p > 0$, however, prioritizing around half of the available resources incurs the greatest penalty in terms of running time.
Recall that each $\Fp \psi$ in $\varphi_{r, r_p}$ is first rewritten to $(p\rightarrow (p\U(\neg p\U\rel(\psi))))\wedge(\neg p\rightarrow (\neg p\U(p\U\rel(\psi))))$ before being given to \bosy, while the traditional $\F\psi$ operator is a shorthand for the significantly smaller formula $\true \U \psi$.
Thus, for increasing $r_p$, the automaton and consequently the formula given to the \qbf solver becomes larger.
We noticed that determining that no realization of $\phi_{r, r_p}$ with some parameters $n$ and $k$ exists was faster for increasing $r_p$, in particular for $r = 5$ and $r = 6$.
Hence, the search terminates earlier despite an increased number of solved \qbf queries, resulting in an overall smaller slowdown.
%For increasing $r_p$, however, the formula $\rel(\varphi_{r, r'})$ becomes more and more homogenous, i.e., its structure becomes more and more regular.
%Hence, the underlying \qbf solver is able to use this homogeneity to speed up the inference process\todo{Are we sure about this?}.
%Thus, for $r_p \approx \nicefrac{r}{2}$, the formula given to the \qbf solver is quite large, without exhibiting enough homogeneity that would speed up the inference, which explains the peak in the slowdown around this number of prioritized resources.

After having evaluated the running time of our tool against that of that of the underlying bounded synthesis tool, we now evaluate the feasibility of our approach for the search for a strategy of a given size realizing a formula with respect to some given bound.
In other words, we are given some $\varphi_{r,r_p}$, some size $n$ and a bound $k$ and want to decide whether or not a strategy of size at most $n$ exists that realizes~$\varphi_{r,r_p}$ with respect to $2k$.

In order to satisfy the requirement that every request is eventually granted, at least $r-1$ states are required.
Also, the smallest possible strategy is a round-robin strategy, which simply grants each resource in order.
This strategy realizes the formula with respect to the bound~$r$.
These two propositions yield upper bounds on $n$ and $k$ for a given $r$.
Hence, for each $r \in [1,10]$ and each $r_p \in [1,r]$ we search for implementations of $\varphi_{r,r_p}$ of size~$n \in [r-1,2r]$ and with respect to the bound~$k \in [1,r]$ on the block-size.

We show the results for $\varphi_{4,1}$ and $\varphi_{6,2}$ in Figure~\ref{fig:eval_runtime_individual}.
Green circles, red squares, and yellow triangles denote realizable parameter combinations, unrealizable ones, and those for which our implementation timed out after 20 minutes, respectively.
Note that in the benchmark of $\varphi_{6,2}$ there are four invocations that ran longer than $20$ seconds and are thus not shown in the diagram.
The searches for strategies of size~$7$ that realize $\varphi_{6,2}$ with respect to the bounds $2$ and $4$, respectively, as well as the search for a strategy of size $8$ that realizes $\varphi_{6,2}$ with respect to the bound $2$ were eventually unsuccessful after 23 seconds, 833 seconds, and 1\thinspace 009 seconds, respectively.
There exists, however, a strategy of size $12$ that realizes $\varphi_{6,2}$ with respect to $2$, which was found after $72$ seconds.
\begin{figure}
	\centering
	\resizebox{.45\textwidth}{!}{\begin{tikzpicture}[thick,scale=.85]
		\begin{axis}[
			thick,
			title={$4$ resources, $1$ prompt},
			xlabel=$n$,
			xtick=data,
			ylabel=$k$,
			ytick=data,
			grid=both,
			view={-18}{25},
			zlabel={Running time $[\mathit{sec}]$},
			colormap={MYN}{rgb=(1,1,0); rgb=(1,0,0); rgb=(0,1,0)},
			point meta=\thisrow{color},
			point meta min=-1,
			point meta max=1,
			enlarge x limits=true,
			enlarge y limits=true,
			zmin=0,zmax=1,
			visualization depends on={value \thisrow{shape} \as \myshape},
			scatter/@pre marker code/.append style={/tikz/mark=\myshape},
			mark options={scale=1.5}
		]
			\addplot3[surf,only marks,ycomb,scatter] table {data/tikz-4-1.dat};
		\end{axis}
	\end{tikzpicture}}
	\hfill
	\resizebox{.45\textwidth}{!}{\begin{tikzpicture}[thick,scale=.85]
		\begin{axis}[
			thick,
			title={$6$ resources, $2$ prompt},
			xlabel=$n$,
			xtick=data,
			ylabel=$k$,
			ytick=data,
			grid=both,
			view={-18}{25},
			zlabel={Running time $[\mathit{sec}]$},
			colormap={MYN}{rgb=(1,1,0); rgb=(1,0,0); rgb=(0,1,0)},
			point meta=\thisrow{color},
			point meta min=-1,
			point meta max=1,
			enlarge x limits=true,
			enlarge y limits=true,
			zmax=20,
			zmin=0,
			visualization depends on={value \thisrow{shape} \as \myshape},
			scatter/@pre marker code/.append style={/tikz/mark=\myshape},
			mark options={scale=1.5}
		]
			\addplot3[surf,only marks,ycomb,scatter] table {data/tikz-6-2.dat};
		\end{axis}
	\end{tikzpicture}}
	\caption{Running times and results for $\varphi_{4,1}$ on the left-hand side and $\varphi_{6,2}$ on the right-hand side. Green circles denote realizable parameters, while red squares denote unrealizable parameters. Yellow triangles denote that the tool encountered a time-out after 20 minutes.}
	\label{fig:eval_runtime_individual}
\end{figure}
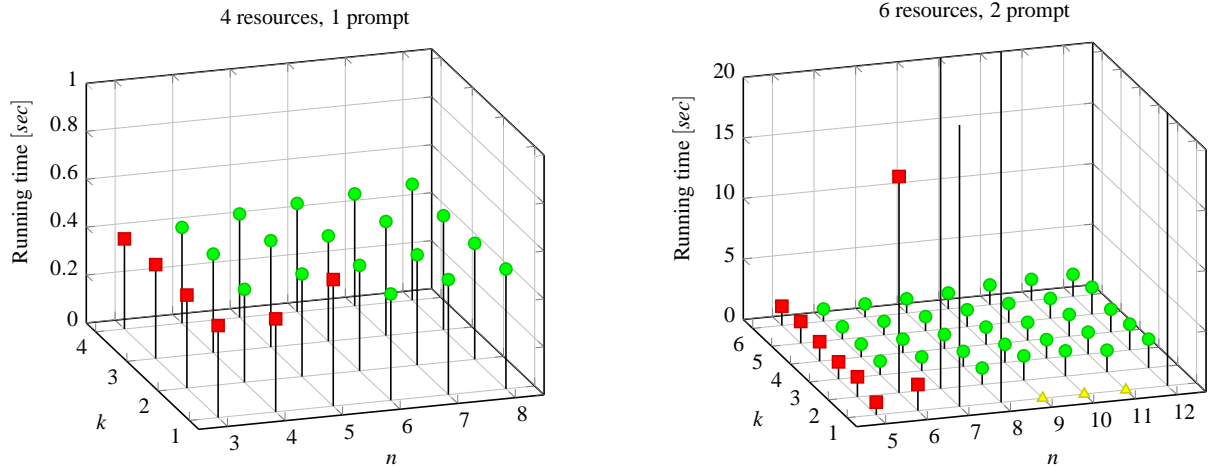

Note that both evaluations shown in Figure~\ref{fig:eval_runtime_individual} exhibit a tradeoff.
There exist strategies that realize $\phi_{6,2}$ with respect to the bounds $6$, $4$, and $2$.
These strategies have size $6$, $8$, and $12$, respectively.
We show the minimal strategies $\sigma_{6,3}$ and $\sigma_{12,1}$ realizing $\varphi_{6,2}$ with respect to the bounds $6$ and $2$, respectively, in Figure~\ref{fig:eval_strategies}. The strategy~$\sigma_{6,3}$ proceeds in a round-robin fashion using only $6$ states while $\sigma_{12,1}$ 
grants $p_1$ every second step using $12$ states to ensure that all requests are eventually granted.

\begin{figure}
	\centering
	\begin{tikzpicture}[thick,>=stealth,node distance=.7 and 1.5]
		\node[draw,shape=circle,initial] (6-3-1) {};
		\node[draw,shape=circle,above=of 6-3-1] (6-3-2) {};
		\node[draw,shape=circle,right=of 6-3-2] (6-3-3) {};
		\node[draw,shape=circle,below=of 6-3-3] (6-3-4) {};
		\node[draw,shape=circle,below=of 6-3-4] (6-3-5) {};
		\node[draw,shape=circle,left=of 6-3-5] (6-3-6) {};
		
		\path[draw,->] (6-3-1) -- node[anchor=east] {$\nicefrac{2^I}{\set{p_2}}$} (6-3-2);
		\path[draw,->] (6-3-2) -- node[anchor=south] {$\nicefrac{2^I}{\set{p_6}}$} (6-3-3);
		\path[draw,->] (6-3-3) -- node[anchor=west] {$\nicefrac{2^I}{\set{p_4}}$} (6-3-4);
		\path[draw,->] (6-3-4) -- node[anchor=west] {$\nicefrac{2^I}{\set{p,p_1}}$} (6-3-5);
		\path[draw,->] (6-3-5) -- node[anchor=north] {$\nicefrac{2^I}{\set{p,p_5}}$} (6-3-6);
		\path[draw,->] (6-3-6) -- node[anchor=east] {$\nicefrac{2^I}{\set{p,p_3}}$} (6-3-1);
		
		\node[anchor=south east] at (6-3-2.north west) {$\sigma_{6,3}$:};
		
		\begin{scope}[node distance=0cm and 3cm]
			\node[draw,shape=circle,initial,right=of 6-3-4] (12-1-1) {};
		\end{scope}
		
		\node[draw,shape=circle,above=of 12-1-1] (12-1-2) {};
		\node[draw,shape=circle,right=of 12-1-2] (12-1-3) {};
		\node[draw,shape=circle,right=of 12-1-3] (12-1-4) {};
		\node[draw,shape=circle,right=of 12-1-4] (12-1-5) {};
		\node[draw,shape=circle,right=of 12-1-5] (12-1-6) {};
		\node[draw,shape=circle,below=of 12-1-6] (12-1-7) {};
		\node[draw,shape=circle,below=of 12-1-7] (12-1-8) {};
		\node[draw,shape=circle,left=of 12-1-8] (12-1-9) {};
		\node[draw,shape=circle,left=of 12-1-9] (12-1-10) {};
		\node[draw,shape=circle,left=of 12-1-10] (12-1-11) {};
		\node[draw,shape=circle,left=of 12-1-11] (12-1-12) {};
		
		\path[draw,->] (12-1-1) -- node[anchor=east] {$\nicefrac{2^I}{\set{p,p_2}}$} (12-1-2);
		\path[draw,->] (12-1-2) -- node[anchor=south] {$\nicefrac{2^I}{\set{p_4}}$} (12-1-3);
		\path[draw,->] (12-1-3) -- node[anchor=south] {$\nicefrac{2^I}{\set{p,p_1}}$} (12-1-4);
		\path[draw,->] (12-1-4) -- node[anchor=south] {$\nicefrac{2^I}{\set{p_2}}$} (12-1-5);
		\path[draw,->] (12-1-5) -- node[anchor=south] {$\nicefrac{2^I}{\set{p,p_5}}$} (12-1-6);
		\path[draw,->] (12-1-6) -- node[anchor=west] {$\nicefrac{2^I}{\set{p_1}}$} (12-1-7);
		\path[draw,->] (12-1-7) -- node[anchor=west] {$\nicefrac{2^I}{\set{p,p_2}}$} (12-1-8);
		\path[draw,->] (12-1-8) -- node[anchor=north] {$\nicefrac{2^I}{\set{p_3}}$} (12-1-9);
		\path[draw,->] (12-1-9) -- node[anchor=north] {$\nicefrac{2^I}{\set{p,p_1}}$} (12-1-10);
		\path[draw,->] (12-1-10) -- node[anchor=north] {$\nicefrac{2^I}{\set{p_2}}$} (12-1-11);
		\path[draw,->] (12-1-11) -- node[anchor=north] {$\nicefrac{2^I}{\set{p,p_6}}$} (12-1-12);
		\path[draw,->] (12-1-12) -- node[anchor=east] {$\nicefrac{2^I}{\set{p_1}}$} (12-1-1);
		
		\node[anchor=south east] at (12-1-2.north west) {$\sigma_{12,1}$:};
		
	\end{tikzpicture}
	\caption{Two strategies $\sigma_{6,3}$ and $\sigma_{12,1}$ realizing $\varphi_{6,2}$ with respect to the bounds $6$ and $2$, respectively. A transition of the form $m \xrightarrow{\nicefrac{i}{o}} m'$ denotes that upon reading $i \in 2^I$ in state $m$, Player~$O$ outputs $o \in 2^O$ and updates her memory to $m'$ (cf.\ Subsection~\ref{subsection_realizability}).}	
	\label{fig:eval_strategies}
\end{figure}
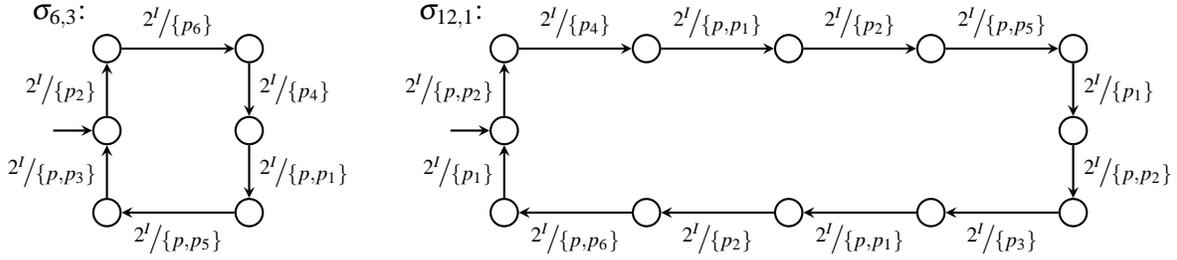

We also see that, in general, unsuccessful searches for a strategy with given size and bound take longer than successful searches for larger strategies or for strategies with a larger bound.
Intuitively, this is due to the fact that the resulting \qbf formula is satisfiable if, and only if, there exists a strategy with the given parameters and that refuting all possible strategies of size $n$ is, in general, harder than showing that such a strategy exists.
Hence, it is of interest to to investigate the border between the realizable and the unrealizable parameters.
We do so in the next section.

%%%%%%%%%%%%%%%%%%%%%%%%%%%%%%%%%%%%%%%%%%%%%%%%%%%%%
%%%%%%%%%%%%%%%%%%%%%%%%%%%%%%%%%%%%%%%%%%%%%%%%%%%%%
\section{Trading Memory for Quality and Vice Versa}
\label{section_tradeoffs}
We have seen in the previous section that there exist $\prompt$ formulas $\varphi$ that exhibit a tradeoff, i.e., for some $k < k'$, the minimal strategy realizing $\varphi$ with respect to $2k$ may be larger than the minimal strategy realizing $\varphi$ with respect to $2k'$.
In this section, we investigate the Pareto frontier of this tradeoff, i.e., those positions in the search space shown in the previous section, at which it is not possible to decrease either the size of the strategy or the bound it realizes without increasing the other value.
To this end, we define the set of realizable parameters $\realparms(\varphi) \subseteq \nats \times \nats$ of $\varphi$ such that $(n, k) \in \realparms(\varphi)$ if and only if there exists a strategy $\sigma$ with $\card{\sigma} = n$ that satisfies $\varphi$ with respect to $k$. Note that $\realparms(\varphi)$ is upwards-closed, i.e., if $(n, k) \in \realparms(\varphi)$, then also $(n+1, k) \in \realparms(\varphi)$ and $(n, k+1) \in \realparms(\varphi)$.

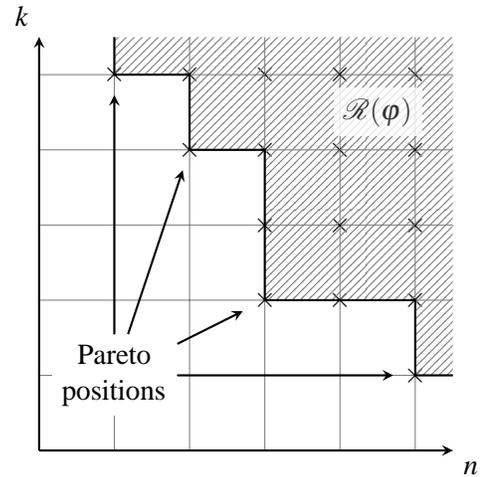
\begin{wrapfigure}{r}{.40\textwidth}
	\centering
	\begin{tikzpicture}[thick]
	\newcommand{\assignmentspace}[1]{
		\path[draw,thin,gray] (0,0) grid (#1.5,#1.5);
		\path[draw,-stealth] (0,0) -- (#1.5,0);
		\path[draw,-stealth] (0,0) -- (0,#1.5);
		\node[anchor=north west] at (#1.5,0) {$n$};
		\node[anchor=south east] at (0,#1.5) {$k$};
	}
	\assignmentspace{5}
	
	% y = 5
	\foreach \x in {1,...,5} \node (\x-5) at (\x,5) {\textbf{$\times$}};
	% y = 4
	\foreach \x in {2,...,5} \node (\x-4) at (\x,4) {\textbf{$\times$}};
	% y = 3
	\foreach \x in {3,...,5} \node (\x-3) at (\x,3) {\textbf{$\times$}};
	% y = 2
	\foreach \x in {3,...,5} \node (\x-2) at (\x,2) {\textbf{$\times$}};
	% y = 1
	\foreach \x in {5,...,5} \node (\x-1) at (\x,1) {\textbf{$\times$}};
	
	\node[fill=white,align=center] (tradeoffs-label) at (1,1) {Pareto\\positions};
	\foreach \x/\y in {1/5,2/4,3/2,5/1}
		\path[draw,-stealth] (tradeoffs-label) to (\x-\y);
		
	\path[pattern=north east lines,pattern color=gray] (1,5.5) -- (1,5) -- (2,5) -- (2,4) -- (3,4) -- (3,2) -- (5,2) -- (5,1) -- (5.5,1) -- (5.5,5.5) -- (1,5.5);
	\path[draw] (1,5.5) -- (1,5) -- (2,5) -- (2,4) -- (3,4) -- (3,2) -- (5,2) -- (5,1) -- (5.5,1);

	\node[fill=white,opacity=.8] (realparms-label) at (4.5,4.5) {$\realparms(\varphi)$};
	\end{tikzpicture}
	\caption{The geometrical interpretation of $\realparms(\varphi)$ and the Pareto positions of $\varphi$.}
	\label{fig:tradeoff_geom}
\end{wrapfigure}
A Pareto position of a formula $\varphi$ is a pair of realizable parameters $(n, k) \in \realparms(\varphi)$ such that it is not possible to realize the bound $k$ with a strategy of size $n-1$, and no strategy of size $n$ realizes a smaller bound than $k$.
Formally, a pair of realizable parameters $(n, k) \in \realparms(\varphi)$ is a Pareto position if both $(n-1, k) \notin \realparms(\varphi)$ and $(n, k-1) \notin \realparms(\varphi)$.
When considering the set $\realparms(\varphi)$ geometrically, the Pareto positions of $\varphi$ are the corner points of the area defined by $\realparms(\varphi)$, as shown in Figure~\ref{fig:tradeoff_geom}.

By a simple geometrical argument over the space $\nats \times \nats$ that combines Theorem~\ref{thm_prompt} with the upwards-closure of $\realparms(\varphi)$ we obtain a doubly-exponential bound in $\card{\phi}$ on the number of Pareto positions of~$\phi$.

\begin{lemma}
Let $\varphi$ be a $\prompt$ formula.
There exist at most $\bigo(2^{2^{\card{\varphi}}})$ Pareto positions of $\varphi$.
\end{lemma}

\begin{proof}
If $\varphi$ is not realizable with respect to any bound $k$, then we have $\realparms(\varphi) = \emptyset$ and thus, the statement holds true.
Thus, assume $\varphi$ is realizable with respect to some $k$.
Due to Theorem~\ref{thm_prompt}, we obtain that $\varphi$ is realizable with respect to some $k' \in \bigo(2^{2^{\card{\varphi}}})$, which is witnessed by a strategy of size $n' \in \bigo(2^{2^{\card{\varphi}}})$, i.e., $(n', k') \in \realparms(\varphi)$.

Clearly, there are at most $k'$ Pareto positions $(n, k)$ with $k \leq k'$, since otherwise, upwards-closure of $\realparms(\varphi)$ would be violated.
For the same reason, there are at most $n'$ Pareto positions $(n, k)$ with $n \leq n'$.
Finally, there can exist no Pareto positions $(n, k)$ with either $n \geq n'$ or $k \geq k'$, again due to upwards-closure of $\realparms(\varphi)$.
Thus, there exist at most $n' + k' \in \bigo(2^{2^{\size{\phi}}})$ Pareto positions of $\varphi$.
\end{proof}

Having shown that the number of Pareto positions has an upper bound, we now investigate the Pareto frontier in general.
We show that fixing one parameter yields exponential and doubly-exponential upper bounds on the other parameter, respectively.
For a fixed $n$, this upper bound on $k$ is obtained by a reduction to the model checking problem for $\prompt$.
For a fixed $k$, however, we obtain the upper bound on $n$ by turning $\phi$ into a parity game of doubly-exponential size and solving this game.

\begin{lemma}
    Let $\varphi$ be a $\prompt$ formula.
\begin{enumerate}
  \item \label{lemma_tradeoffs_nfixed}
 Let $\sigma$ be a strategy that realizes $\varphi$ with $\card{\sigma} = n$.
    Then $(n, k) \in \realparms(\varphi)$ for some $k \in \bigo(n \cdot 2^{\card{\varphi}})$.
  
  \item \label{lemma_tradeoffs_kfixed} 
   Let $\varphi$ be realizable w.r.t.\ $k$.
    Then $(n, k) \in \realparms(\varphi)$ for some $n \in \bigo(2^{\card{\varphi}^2 \cdot (2(k+1))^{2 \card{\varphi}}})$ %$m \in \bigo(2^{2^{(k+1) \cdot \card{\varphi}}})$.
\end{enumerate}
\end{lemma}
\begin{proof}
    \ref{lemma_tradeoffs_nfixed}.)
    Fixing a strategy $\sigma$ of size $n$ simplifies the realizability problem to the problem of model checking $\prompt$.
    The upper bound of $k$ in the model checking problem for $\pltl$, which includes $\prompt$, is known to be linear in $n$ and exponential in $\card{\varphi}$~\cite{AlurEtessamiLaTorrePeled01}. % rational: in Kupferman et al., the complexity is quadratic in $n$

    \ref{lemma_tradeoffs_kfixed}.)
    Given a bound $k$, we can translate $\varphi$ to a parity game $\mathcal{P}$ of size $\bigo(2^{\card{\varphi}^2 \cdot (2(k+1))^{2 \card{\varphi}}})$ that is winning for player 0 if, and only if, $\varphi$ is realizable with bound $k$~\cite{Zimmermann13}.
    As a positional winning strategy for player 0 in $\mathcal{P}$ can be translated into a realizing strategy $\sigma$ for $\varphi$ with respect to $k$, with $\card{\sigma} \in \bigo(\card{\mathcal{P}})$. This proves our upper bound on the size of a realizing strategy.
    %Given a bound $k$, we can translate $\varphi$ to an $\ltl$ formula $\varphi'$ such that $\varphi$ can be satisfied with bound $k$ iff $\varphi'$ can be satisfied.
    %This construction replaces subformulas $\Fp \psi$ by a chain of Next operators $\psi \lor \X(\psi \lor \X(\dots \psi))$ of length $k$.
    %The size of $\varphi'$ (measured in terms of distinct subformulas) is in $\bigo((k+1)\cdot\card{\varphi})$ and the upper bound on $m$ follows from the realizability problem for $\ltl$~\cite{PnueliRosner89}.
\end{proof}

The previous two lemmas each presented upper bounds on the number of Pareto positions.
These bounds permit us to restrict the search space when looking for a realizing strategy:
Instead of fixing some $n$ or $k$ and checking doubly-exponentially many possibilities for the respective other parameter, we only need to consider exponentially many possible values for it.

We now turn our attention to the respective lower bounds, i.e., we provide a family of formulas $\phi_b$ that exhibit such a Pareto frontier.
More precisely, for each $\varphi_b$, there exists a family of strategies $\sigma_{b,j}$ such that $\sigma_{b,j}$ is of size exponential in $j$ and realizes $\varphi_b$ with respect to some $k$ that is exponential in $b-j$.
Each of these $\sigma_{b,j}$ is minimal for its respective bound.

Intuitively, $\varphi_b$ describes a game in which Player~$O$ decides at the beginning how much memory she wants to use by playing some number $j$.
Player~$I$ then plays some number in $[0;2^j)$, which Player~$O$ has to repeat afterwards, thus requiring her to use exponential memory in $j$.
Afterwards, Player~$I$ implements a binary counter using $b-j$ bits.
The game ends once Player~$I$ has counted up to $2^{b-j}-1$.
Moreover, $\varphi_b$ requires that this end is reached promptly, i.e., the bound $k$ is in $\bigo(2^{b-j})$, while every strategy realizing $\phi_b$ with respect to that bound $k$ has at least size~$2^j$.

\begin{theorem}
\label{thm:continuous_tradeoff}
For each $b \in \nats$ there exists a $\prompt$ formula $\varphi_b$ with $\card{\varphi_b} \in \bigo(b)$ such that for each $0 \leq j \leq b$, there exists an $n \in \bigo(2^j)$ and a $k \in \bigo(2^{b-j})$, such that $(n, k)$ is a Pareto position.
\end{theorem}

\begin{proof}
	We construct an $\ltl$ formula $\varphi_b$ that specifies the following game $\game_b$ for a given $b$ with $P = I \cup O$, where $I = \set{\imark, \idelim}$ and $O = \set{\omark, \odelim}$.

	The game begins with Player~$O$ playing some number~$0 \leq j \leq b$ in unary encoding, i.e., she plays~$j$ times her proposition \omark and ends this encoding by playing~\odelim.
	After this first \odelim, Player~$I$ plays the binary encoding of some number $0 \leq n < 2^j$ using~$j$ positions, and finishes with a \idelim.
	After Player~$I$ has issued his \idelim, Player~$O$ must repeat his sequence and finish with \odelim.
	When Player~$O$ has finished repeating Player~$I$'s sequence, Player~$I$ must implement a binary counter with $b - j$ bits, starting with the binary encoding of $0$.
	Two consecutive values of the counter must be delimited by \idelim, and after encoding $2^{b-j} - 1$, Player~$I$ must play $\idelim\idelim$. 
	During the respective other player's turn, both players always have to play the empty set.
	If either player does not conform to the rules of this game, she loses.
	A play of this game is illustrated in Figure~\ref{fig:game_tradeoffs}.
	\begin{figure}
	\centering
	\begin{tikzpicture}[thick,yscale=.9]
		\node (player-1-label) at (0,0) {Player $I$};
		\node (player-0-label) at (0,-1) {Player $O$};
		
		\node (left) at (1,-1) {\omark};
		\node at (1.5,-1) {$\cdots$};
		\node (right) at (2,-1) {\omark};
		\node at (2.5,-1) {\odelim};
		
		\path[draw] (1,0) -- (2.5,0);
		
		\draw[decorate,decoration={brace,amplitude=3pt}] (right.south) -- node [anchor=north] {$j$ times} (left.south);
		
		\node (left) at (3,0) {\imark};
		\node at (3.5,0) {$\emptyset$};
		\node at (4,0) {$\cdots$};
		\node at (4.5,0) {\imark};
		\node (right) at (5,0) {$\emptyset$};
		\node at (5.5,0) {\idelim};
		
		\path[draw] (3,-1) -- (5.5,-1);
		
		\draw[decorate,decoration={brace,amplitude=3pt}] (left.north) -- node [anchor=south] {$j$ positions} (right.north);
		
		\node (left) at (6,-1) {\omark};
		\node at (6.5,-1) {$\emptyset$};
		\node at (7,-1) {$\cdots$};
		\node (right) at (7.5,-1) {\omark};
		\node (right) at (8,-1) {$\emptyset$};
		\node at (8.5,-1) {\odelim};
		
		\path[draw] (6,0) -- (8.5,0);
		
		\draw[decorate,decoration={brace,amplitude=3pt}] (right.south) -- node [anchor=north] {$j$ positions} (left.south);
		
		\node (left) at (9,0) {$\emptyset$};
		\node at (9.5,0) {$\cdots$};
		\node (right) at (10,0) {$\emptyset$};
		\node at (10.5,0) {\idelim};
		
		\draw[decorate,decoration={brace,amplitude=3pt}] (left.north) -- node [anchor=south] {$b-j$ positions} (right.north);
		
		\node at (11,0) {$\cdots$};
		\node at (11.5,0) {\idelim};
		\node (left) at (12,0) {\imark};
		\node at (12.5,0) {$\cdots$};
		\node (right) at (13,0) {\imark};
		\node at (13.5,0) {\idelim};
		\node at (14,0) {\idelim};
		
		\draw[decorate,decoration={brace,amplitude=3pt}] (left.north) -- node [anchor=south] {$b-j$ positions} (right.north);
		
		\node (left) at (9,-1) {\textcolor{white}{$\emptyset$}};
		\node (right) at (13.5,-1) {\textcolor{white}{\odelim}};
		\path[draw] (9,-1) -- (14,-1);
		
%		\foreach \x in {8.5,9,...,11.5}
%			\node at (\x,-1) {$\emptyset$};
			
		\draw[decorate,decoration={brace,amplitude=3pt}] (right.south) -- node [anchor=north] {$(b-j+1)\cdot(2^{b-j})$ positions} (left.south);
		
		\node at (14.5,0) {$\cdots$};
		\node at (14.5,-1) {$\cdots$};
		
%		\node at (0,-2.5) {Position};
%		\node at (1,-2.5) {$0$};
%		\node at (2.5,-2.5) {$j$};
%		\node at (5.5,-2.5) {$2j+1$};
%		\node at (8.5,-2.5) {$3j+2$};
%		\node at (13.5,-2.5) {$3j+2 + (b-j+1)\cdot(2^{b-j})$};

	\end{tikzpicture}
	\caption{A play of the game $\game_b$. Sequences of $\emptyset$ are denoted by black lines for readability.}
	\label{fig:game_tradeoffs}	
	\end{figure}
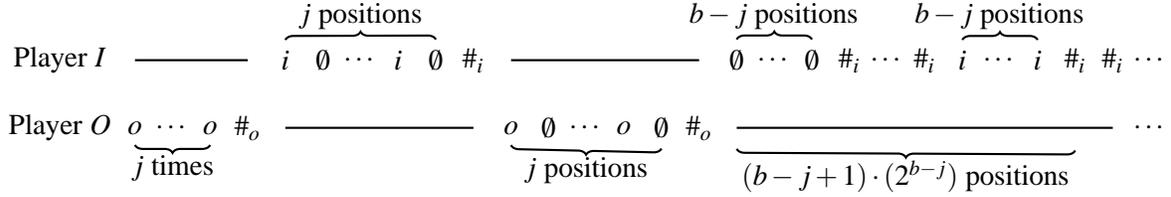

	Towards a formal definition of $\varphi_b$, fix some $j$ with $0 \leq j \leq b$.
	We construct formulas $\varphi^\pick_{b,j}$, $\varphi^I_{b,j}$ and $\varphi^O_{b,j}$ that encode the fact that Player~$O$ starts by announcing the number $j$ in unary encoding, and assumptions about the behavior of both players in $\game_{b}$, respectively, if Player~$O$ starts by announcing~$j$.

	The formula $\varphi^\pick_{b,j}$ is trivial to construct in linear size in $j$ using nested $\X$-operators, as it just argues about a prefix of length $j+1$ of the resulting play.
	The formula $\varphi^I_{b,j}$ encodes the following assumptions about the behavior of Player~$I$:
	\begin{enumerate}
		\item At any time, Player~$I$ picks either \imark or \idelim, or neither, but never both,
		\item Player~$I$ plays $\emptyset$ until Player $O$ plays \odelim for the first time,
		\item immediately after the first position where Player~$O$ plays $\odelim$, Player~$I$ does not pick \idelim for $j$ positions, but does pick it after $j$ turns,
		\item after Player~$I$ has played \idelim for the first time, he plays $\emptyset$ until Player~$O$ has played \odelim again,
		\item immediately after Player~$O$ has played \odelim for the second time, Player~$I$ plays $\emptyset$ for $b-j$ turns, followed by \idelim,
		\item whenever Player~$I$ plays \idelim after the first time he has done so, if the $b-j$ positions preceding that \idelim encode some $\ell \in [0, 2^{b-j}-1)$ in binary using the proposition~\imark, the $b-j$ positions succeeding that \idelim encode $\ell+1$ and are followed directly by another \idelim, and
		\item if and when Player~$I$ encodes $2^{b-j}-1$ at some point after he has played his first \idelim, this encoding is directly followed by $\idelim\idelim$.
	\end{enumerate}
	These properties can be specified using polynomial-length $\ltl$ formulas in both $b$ and $j$.
	In particular, the correct behavior of the $(b-j)$-bit counter can be specified using a formula of polynomial size in $(b-j)$ using standard constructions.
	Thus, we obtain the formula $\varphi^I_{b,j}$ of polynomial size in $b$ and $j$.
	
	Similarly, the formula $\varphi^O_{b, j}$ encodes the following guarantees that Player~$O$ has to ensure, if the assumptions regarding the play of Player~$I$ are met:
	\begin{enumerate}
		\item At any time, Player~$O$ plays either \omark or \odelim, or neither, but never both,
%		\item Player~$O$ starts the game by playing $O$ for $j$ times, directly followed by a \odelim,
		\item after playing \odelim for the first time, Player~$O$ only plays $\emptyset$, until Player~$I$ plays a \idelim,
		\item if Player~$I$ has played some word $w \in (\set{\imark} + \emptyset)^j$ directly preceding his first \idelim, then Player~$O$ must play $w$ with every \imark replaced by an \omark immediately after Player~$I$ has played his first \idelim, and
		\item after her second \odelim, Player~$O$ exclusively plays $\emptyset$.
	\end{enumerate}
	Again, all these properties only argue about infixes of linear size in $j$ and can all be specified using formulas of polynomial size in $b$ and $j$.
	Hence, we obtain a formula $\varphi^O_{b, j}$ that specifies all these guarantees, which is again of size polynomial in $b$ and $j$.
	
	Using the formulas $\varphi^\pick_{b,j}$, $\varphi^I_{b, j}$, and $\varphi^O_{b,j}$ we then define $\varphi_b = \bigvee_{0 \leq j \leq b} \varphi^\pick_{b,j} \land (\varphi^I_{b,j} \rightarrow \varphi^O_{b, j} \land \Fp( \idelim \land \X \idelim)),$ which formally denotes the requirement that Player~$O$ starts by playing some $j$ in unary encoding, and, if Player~$I$ satisfies the assumptions about his behavior in this situation, then Player~$O$ fulfills the requirements to her part of the play, and that Player~$I$ promptly plays $\idelim\idelim$, i.e., promptly finishes counting.
	
	We now show that for each $j$ in $[0;b]$, there exist an $n \in \bigo(2^j)$ and a $k \in \bigo(2^{b-j})$ such that $(n, k)$ is a Pareto position.
	To this end, fix some $j$ with $0 \leq j \leq b$.
	% $\card{\sigma} = 3 \cdot 2^j - 3$
	%	We show the memory structure implementing this strategy in Figure~\ref{fig:memory_struct} as a Mealy automaton, i.e., each transition is labeled with a pair $\nicefrac{P}{P'}$, which denotes that, if Player~$I$ plays $P$ in the current state, Player~$O$ plays $P'$ and updates the memory to the state at the end of the transition.
	Clearly, Player~$O$ has a strategy $\sigma_{b,j}$ with $\card{\sigma_{b,j}} \in \bigo(2^j)$ that realizes $\varphi$ with respect to some $k \in \bigo(2^{b-j})$.
	Intuitively, Player~$O$ first uses $j+1$ memory states as a unary counter up to $j$.
	She plays \omark until this counter reaches $j$, which happens after $\bigo(j)$ steps.
	Once the counter has reached $j$, Player~$O$ plays \odelim and stores the number encoded by Player~$I$ using $\bigo(2^j)$ memory states, which again takes $\bigo(j)$ steps.
%	Here, she may reuse the states used to implement the counter in the first part of $\game_b$, as she can distinguish the different phases by the choices of Player~$I$.
	After Player~$I$ has played \idelim, Player~$O$ repeats the encoding of the number played by Player~$I$, again using $\bigo(2^j)$ memory states and $\bigo(j)$ steps.	
	Afterwards, Player~$O$ plays a single \odelim followed by $\emptyset$ ad infinitum. %, using the assumptions about Player~$I$'s behavior.
	Player~$I$ then implements a binary counter with $b-j$ bits and has to play $\idelim\idelim$ after that counter has reached its maximal value.
	This occurs after $\bigo((b-j)\cdot 2^{b-j})$ steps.
%	Hence, this strategy requires $(2^{j+1} -1) + (2^j -1)  = 3 \cdot 2^j -j 3 \in \bigo(2^j)$ memory states and realizes $\varphi_b$ with respect to $\bigo(2^{b-j})$.
	Hence, this strategy requires $\bigo(j + 2^j + 2^j) = \bigo(2^j)$ memory states and realizes $\varphi_b$ with respect to some $k \in \bigo(3j + (b-j) \cdot 2^{b-j}) = \bigo(2^{b-j})$.

	It remains to show that $\bigo(2^j)$ is a lower bound on the size of any strategy that realizes some bound $k \in \bigo(2^{b-j})$ and that $\bigo(2^{b-j})$ is a lower bound on the parameter $k$ with respect to which a strategy with size in $\bigo(2^j)$ can realize $\phi_b$.
	First, assume that there exists a strategy $\sigma'_{b,j}$ with $\card{\sigma'_{b,j}} \in \smallo(2^j)$ that realizes $\varphi_b$ with respect to some $k \in \bigo(2^{b-j})$.
	Then, there must exist two numbers $0 \leq \ell < \ell' < 2^j$ such that $\sigma'_{b, j}$ ends up in the same state after the two plays $\omark^j\odelim\textsc{bin}_j(\ell)\idelim$ and $\omark^j\odelim\textsc{bin}_j(\ell')\idelim$, where $\textsc{bin}_j(\ell)$ denotes the encoding of $\ell$ in binary using $j$ bits (encoded by $\imark$).
	Hence, Player~$O$ cannot differentiate between $\ell$ and $\ell'$ and does not ensure her guarantees in one of the two cases.
	Thus, $\sigma'_{b,j}$ does not realize $\varphi_b$.
		
	Moreover, it is clear that, due to the strict structure of the game, Player~$O$ cannot force the occurrence of $\idelim\idelim$ in $\smallo(2^{b-j})$ steps using a memory structure of size $\bigo(2^j)$.
	The only way for her to force Player~$I$ to play $\idelim\idelim$ after less than $\bigo(2^{b-j})$ steps is to play some number $j' > j$ at the beginning of the game.
	Doing so, however, would give Player~$I$ $j'$ bits to encode some number at the beginning of the second part of the game, which in turn would require Player~$O$ to use $\bigo(2^{j'})$ memory states to store and repeat this number, as argued before.
\end{proof}

We observe that each $\phi_b$ has linearly many Pareto positions in $b$, where the extremal values in $\realparms(\phi_b)$ are $(n, k)$ for $n \in \bigo(1)$, $k \in \bigo(2^b)$ and $(n',k')$ for $n' \in \bigo(2^b)$ and $k' \in \bigo(b)$.
In order to show that the distance between $n$ and $k$ may even become doubly-exponential, we move from the continuous tradeoff exhibited by the previous theorem to a discrete tradeoff, i.e., for each $b$ we provide a formula $\phi_b$ such that there are two ways to realize $\phi_b$;
Either, Player~$O$ realizes this formula with respect to some constant bound, but requires doubly-exponential memory to do so, or she realizes it with respect to some exponential bound, but can do so by using only constant memory.

These bounds are obtained by letting Player~$O$ choose between one of two games, in which Player~$I$ has to implement either a doubly- or singly-exponentially bounded counter.
In the former case, this realization is formalized by an $\ltl$ formula, hence the specification is trivially realized with respect to $k=0$.
Player~$O$ does, however, require doubly-exponential memory to denote errors in Player~$I$'s implementation of the counter~\cite{PnueliRosner89}.
In the latter case, Player~$O$ does not require any memory, but the specification requires that Player~$I$ finishes counting promptly.
Hence, the specification is only fulfilled with respect to an exponential bound.

\begin{theorem}
  For each $b \in \nats$ there exists a $\prompt$ formula $\varphi_b$ with $\card{\varphi_b} \in \bigo(b)$ such that there exist $n \in \bigo(2^{2^b})$, $n' \in \bigo(1)$, and $k' \in \bigo(2^b)$, such that both $(n, 0)$ and $(n', k')$ are Pareto positions of $\phi_b$.
\end{theorem}
\begin{proof}
  Let $b \in \nats$.
  We give a realizable $\prompt$ formula $\varphi_b$ that exhibits the stated tradeoff.
  Let $\psi$ be a realizable $\ltl$ formula with $\card{\psi} \in \bigo(b)$ where each strategy realizing $\psi$ has at least doubly-exponentially many states in~$b$.
  Let $\psi'$ be a $\prompt$ formula with $\card{\psi'} \in \bigo(b)$ that is realizable with respect to $k \in O(2^{\card{\psi'}})$ and constant strategy size.
  We construct $\varphi$ to be $(o \rightarrow \X\psi) \land (\neg o \rightarrow \X\psi')$ where $o$ is a fresh atomic proposition controlled by Player~$O$.
  Player~$O$ decides in the first step whether she wants to satisfy the $\ltl$ formula $\psi$ or the $\prompt$ formula $\psi'$.
  Given $\psi$ and $\psi'$, it is trivial to verify that the stated properties hold.
  
  It remains to show that such formulas $\psi$ and $\psi'$ exist.
  It is known that a $\ltl$ formula $\psi$ with the required properties exists~\cite{PnueliRosner89}.
  Intuitively, $\psi$ requires Player~$I$ to implement a binary counter with exponentially many bits in $b$, which counts up to $2^{2^b}$.
  The task of Player~$O$ is to mark errors in Player~$I$'s implementation of the counter, for which she requires doubly-exponential memory in $b$.
  
  The $\prompt$ formula $\psi'$ requires Player~$I$ to implement a binary counter, similarly to the latter phase of the game $\game_b$ constructed in the proof of Theorem~\ref{thm:continuous_tradeoff}.
  After Player~$I$ has counted up to $2^b$, he plays some delimiter $\#$.
  Then the formula $\psi'$ is of the form $\psi_\textrm{count} \rightarrow \Fp\#$, where $\psi_\textrm{count}$ specifies the assumption that Player~$I$ implements the binary counter correctly and finishes with a $\#$.
  Clearly, Player~$O$ can realize this formula with a  strategy of size one, but she cannot enforce a realization with respect to some bound $k \in \smallo(2^b)$.
% 
%  Intuitively, the formula $\psi'$ describes a game where the environment player outputs a sequence of a binary counter.
%  We parameterize $\psi'$ with a bound $b$ on the number of counter bits.
%  Given a $\ltl$ formula $\psi_\mathit{counter}$ that encodes the counter, the formula $\psi'$ is $\psi_\mathit{counter} \rightarrow \Fp (\psi_\mathit{error} \lor \psi_\mathit{overflow})$, i.e., the system player wins if the environment makes an error while counting or the counter overflows.
%  All three $\ltl$ formulas are linear in the number of bits $b$ and the environment cannot win the game.
%  However, it can stretch the bound $k$ as long as the counter sequence, i.e., $k \in O(2^b)$.
\end{proof}

%%%%%%%%%%%%%%%%%%%%%%%%%%%%%%%%%%%%%%%%%%%%%%%%%%%%
%%%%%%%%%%%%%%%%%%%%%%%%%%%%%%%%%%%%%%%%%%%%%%%%%%%%
%\section{Beyond Prompt-LTL}
%\label{subsection_beyond}
%\input{beyond}

%%%%%%%%%%%%%%%%%%%%%%%%%%%%%%%%%%%%%%%%%%%%%%%%%%%%%
%%%%%%%%%%%%%%%%%%%%%%%%%%%%%%%%%%%%%%%%%%%%%%%%%%%%%
\section{Conclusion}
\label{section_conclusion}
In this work, we presented an approximation algorithm for the $\prompt$ realizability problem with doubly-exponential running time with an approximation factor of two. This is an exponential improvement over the fastest known exact algorithms. The algorithm relies on repeated calls to an $\ltl$ realizability solver. We have implemented the algorithm using \bosy as $\ltl$ realizability solver, which implements the bounded synthesis approach. In our proof-of-concept experiments, a tradeoff between the size and the quality of a strategy becomes apparent, which we investigated: we proved upper bounds on the tradeoff, which reduces the search space of our algorithm, and proved matching lower bounds.

Although we presented our results only for $\prompt$, they also hold for the more expressive logics~$\pltl$~\cite{AlurEtessamiLaTorrePeled01} and $\pldl$~\cite{FaymonvilleZimmermann14}, as they can be compiled into Büchi automata of exponential size and as the alternating color technique is applicable to them as well.

There are several open problems to consider in future work. Most importantly, the computational complexity of the exact optimization problem is still open. Similarly, the exact memory requirements of optimal strategies are open: triply-exponential memory is always sufficient~\cite{Zimmermann13}, but it is open whether doubly-exponential memory suffices as well, as it does for $\ltl$ specifications. Other open problems relate to the tradeoffs: we have studied the tradeoff between size and quality of strategies. One can also consider tradeoffs between different parameters in $\pltl$ and $\pldl$ formulas or take the running time into account as well. The former problem is tightly related to the study of the solution space, i.e., the space of the realizable parameter valuations (see~\cite{AlurEtessamiLaTorrePeled01} for results on the model checking).

%%%%%%%%%%%%%%%%%%%%%%%%%%%%%%%%%%%%%%%%%%%%%%%%%%%%%
%%%%%%%%%%%%%%%%%%%%%%%%%%%%%%%%%%%%%%%%%%%%%%%%%%%%%
\bibliographystyle{eptcs}
\bibliography{biblio}

\end{document}